\documentclass[aps,pra,twocolumn]{revtex4}
\usepackage{amsfonts}

\usepackage{amsmath}
\usepackage{amssymb}
\usepackage{graphicx}
\usepackage{hyperref}

\newtheorem{theorem}{Theorem}

\newtheorem{definition}[theorem]{Definition}

\newtheorem{lemma}[theorem]{Lemma}
\newtheorem{proposition}[theorem]{Proposition}

\newenvironment{proof}[1][Proof]{\noindent\textit{#1.} }{\ $\Box$}
\DeclareMathOperator{\brank}{rk_{2}}
\newcommand{\ket}[1]{\vert#1\rangle}

\begin{document}

\title{Weight of quadratic forms and graph states}
\author{Alessandro Cosentino}
\affiliation{Dipartimento di Informatica, Universit\`{a} degli Studi di Pisa, 56127 Pisa,
Italy}
\email{cosenal@gmail.com}
\author{Simone Severini}
\affiliation{Department of Physics \& Astronomy, University College London, WC1E 6BT London, United Kingdom}
\date{\today}

\begin{abstract}
We prove a connection between Schmidt-rank and weight of quadratic forms. This provides a new tool for the classification of graph states based on entanglement. Our main tool arises from a reformulation of previously known results concerning the weight of quadratic forms in terms of graph states properties. As a byproduct, we obtain a straightforward characterization of the weight of functions associated with pivot-minor of bipartite graphs.
\end{abstract}

\maketitle

\section{Introduction}

Graph states form a particularly interesting class of multipartite entangled
states associated with combinatorial graphs (see, \emph{e.g.}, \cite%
{Hein2007}), and have applications in diverse areas of quantum information
processing, such as quantum error correction \cite{Schlingemann2007} and the
one-way model \cite{oneway}. A graph state can be represented as a
homogenous quadratic Boolean function, whose monomials correspond to the
edges of the graph. Studying properties of quadratic
Boolean functions has the potential to discover new interesting properties
of graph states and to recast the known properties in a different language.

An important problem concerning graph states is quantifying their
entanglement. The typical approach to this task is to determine equivalence
classes of graph states under local operations (see \cite{Ji2007}, for a
recent work approaching this topic). Besides this, \emph{ad hoc }measures
for quantifying entanglement in graph states have been defined \cite%
{Fattal2007, Eisert2001}.
In order to disprove an important conjecture in the context of equivalence of graph 
states under local operations, the conjecture has been reduced to a simpler problem regarding quadratic Boolean functions \cite{Gross07, Ji2007}.

The correspondence between graph states and Boolean functions has been also studied 
in the context of quantum codes.
The codes with best error-correction capability can be interpreted as describing states with high entanglement. 
Moreover, cryptographic properties of Boolean functions, such as spectral properties,
can be interpreted as entanglement measures of the corresponding quantum states \cite{Parker01, Riera05I, Riera05II}.
Connections between graph states, codes and Boolean functions are 
clearly explained in \cite{DanielsenM}.

The aim of this paper is to take the first steps toward a theory that
relates the weight of quadratic forms with entanglement properties of graph
states. For graph states represented by bipartite graphs, we shall prove a
direct relation between Schmidt rank and the weight of corresponding
quadratic forms. As far as we know, the weight of certain types of
homogeneous quadratic Boolean functions has not been researched thoroughly.
Since this weight corresponds to the number of minus signs in the basis
state decomposition of a graph state, we will call it the \emph{minus sign
number}. Such a number counts the induced subgraphs with an odd number of
edges in the graph associated to the state.

Note that, in a related context, counting the weights of a set of Boolean functions has been proved to be equivalent to determining the output of a quantum computation \cite{Nielsen04}.
Further investigation on the weight of Boolean functions associated with quantum states could then also lead to results in quantum computational complexity.

The remainder of this article is organized as follows. In Sec. \ref%
{sec:basic}, we will recall the basic definitions of graph state and of
quadratic forms. We will also present a polynomial time algorithm,
immediately derived from \cite{ehrenfeucht90computational}, to efficiently
calculate, given a graph, the weight of the corresponding quadratic form. 
We also show how the weight of a Boolean function 
relates to the spectral coefficients of the function.
In Sec. \ref{sec:brank}, we prove a direct relation between Schmidt rank and
weight of quadratic forms for bipartite graphs. Then, in Sec. \ref%
{sec:properties}, we list some properties of the weight of quadratic forms
related to generic graphs. Sec. \ref{sec:examples} contains explicit
expressions to calculate the weight of quadratic forms associated with some
interesting classes of graphs. Finally, in Sec. \ref{sec:conclusion}, we
conclude with a brief discussion of the results and some open problems.

\section{Basic concepts}

\label{sec:basic} It is useful to fix the notation and recall some of the
definitions that we will use in the remainder of the paper.

A \emph{graph} is a pair $G=(V,E)$ whose elements are two sets, $V$ and $%
E\subset \lbrack V]^{2}$. The elements of $V$ are called \emph{vertices}.
The elements of $E$ are unordered pairs of different vertices and are called
\emph{edges}. The \emph{order} and the \emph{size} of a graph are
respectively the number of its vertices, \emph{i.e.} $|V(G)|$, and the
number of its edges,\emph{\ i.e.} $|E(G)|$. The \emph{adjacency matrix} of a
graph $G$ is the matrix $A(G)$ such that $A(G)_{i,j}=1$ if $\{i,j\}\in E(G)$
and $A(G)_{i,j}=0$ if $\{i,j\}\notin E(G)$. A \emph{vertex cover} of a graph
$G$ is a set $S\subseteq V(G)$ such that every edge of $G$ has at least one
of its endpoints in $S$. The \emph{vertex covering number} of a graph $G$ is
the size of its minimal vertex cover. Given a graph $G$, the \emph{%
neighborhood} $N_{i}\subset V(G)$ of a vertex $i\in V(G)$ is the set of all
vertices that are adjacent to $i$. Unless otherwise specified, all our
algebraic operations are performed over $\mathbb{F}_{2}$, the finite field
of order two.

An $n$-qubit \emph{graph state} $|G\rangle $ is a pure state associated with
a graph $G$. Graph states are prepared by assigning each vertex to a qubit
in the state $|+\rangle =(|0\rangle +|1\rangle )/\sqrt{2}$ and applying, for
each edge between two vertices $i$ and $j$, the \emph{Controlled-}$\mathsf{Z}
$ gate on qubits $i$ and $j$:
\begin{equation*}
|G\rangle =\prod_{\{i,j\}\in E(G)}C_{Z_{i,j}}|+\rangle ^{V},
\end{equation*}%
where $C_{Z}=|00\rangle \langle 00|+|01\rangle \langle 01|+|10\rangle
\langle 10|-|11\rangle \langle 11|$.

There is a common way to define graph states in terms of Boolean functions
and quadratic forms. A \emph{Boolean quadratic form} in $n$ variables is a
homogeneous polynomial in $\mathbb{F}_{2}[x_{1},\ldots ,x_{n}]$ of degree $2$%
, or the zero polynomial. We denote by $\mathbf{x}$ the column vector of
variables $x_{1},\ldots ,x_{n}$. Given a graph $G=(V,E)$ of order $n$, we
associate a variable $x_{i}$ with each vertex of the graph. An instance of
the vector $\mathbf{x}$ characterizes a subset $S_{x}$ of $V(G)$, such that $%
i\in S_{x}$ if $x_{i}=1$ and $i\notin S_{x}$, otherwise. Let us associate
with a graph $G$ the following quadratic form:
\begin{equation}
f_{G}(\mathbf{x})=\sum_{\substack{ 1\leq i,j\leq n \\ i<j}}%
A(G)_{i,j}x_{i}x_{j}=\mathbf{x^{T}}U_{G}\mathbf{x},  \label{eq:frommatrix}
\end{equation}%
where $A(G)$ is the adjacency matrix of the graph and $U_{G}$ is its upper
triangular part. A graph $H$ is a \emph{subgraph} of a graph $G$ if $V\left(
H\right) \subseteq V\left( G\right) $, $E\left( H\right) \subseteq E\left(
G\right) $ and every edge in $E\left( H\right) $ has both its vertices in $%
V\left( H\right) $. A subgraph $H$ of $G$ is an \emph{induced} \emph{subgraph%
} if every edge in $E\left( G\right) $, having both vertices in $V\left(
H\right) $, is also in $E\left( H\right) $. Let $G[S_{x}]$ denote the
subgraph induced by the subset $S_{x}\subseteq V(G)$ characterized by $%
\mathbf{x}$. The function $f_{G}(\mathbf{x})=1$ if and only if the size of $%
G[S_{x}]$ is odd.

By using the function $f_{G}$ given in Eq. (\ref{eq:frommatrix}), we can
define the graph state associated with the graph $G$ as a superposition over
all elements of the computational basis for the $n$-qubit space:
\begin{equation}
|G\rangle =\frac{1}{\sqrt{2^{n}}}\sum_{x\in
\{0,1\}^{n}}(-1)^{f_{G}(x)}|x\rangle \,.  \label{eq:deff}
\end{equation}%
The \emph{weight} of a Boolean function $f$ is defined and denoted by $%
|f|=|\{\mathbf{x}\in \{0,1\}^{n}\mid f(\mathbf{x})=1\}|$.

In this paper, we investigate the weight of Boolean quadratic forms with
regards to some properties of graph states. The weight of a quadratic form $%
f_{G}$, associated with a graph $G$, is equal to the number of induced
subgraphs of $G$ with an odd number of edges. In a graph state $|G\rangle $,
the weight of $f_{G}$ corresponds to the number of amplitudes with negative
sign associated with the computational basis vectors. We denote this number
by $w(G)$, that is $w(G)=|f_{G}|$, and we call it \emph{Minus-Signs number
of $G$} (for short \emph{MS-number}). We will use indifferently the
expressions \emph{MS-number of a graph $G$} or \emph{MS-number of a graph
state $|G\rangle $}. This is not confusing since $|G\rangle =|G^{\prime
}\rangle $ if and only if $G=G^{\prime }$. Moreover, we will use the
notation $\bar{w}(G)$ to indicate the number of \emph{plus signs} in a graph
states $|G\rangle $. This is the number of induced subgraphs with \emph{even}
size in $G$: $\bar{w}(G)=2^{|V(G)|}-w(G)$. For instance, the graph state
associated with the complete graph $K_{3}$ is $|K_{3}\rangle =\frac{1}{2%
\sqrt{2}}(|000\rangle +|001\rangle +|010\rangle -|011\rangle +|100\rangle
-|101\rangle -|110\rangle -|111\rangle )$. The MS-number of $K_{3}$ is then $%
w(K_{3})=4$.

The weight of a Boolean function $f:\mathbb{F}_{2}^{n} \rightarrow \mathbb{F}_{2}$ 
relates to the zero-order coefficient of the spectrum of $f$ w.r.t. the \textit{Walsh-Hadamard Transform (WHT)}
\footnote{Walsh-Hadamard Transform is also called \textit{Abstract Fourier Transform}.}.
Let $\mathbf{f}$ be the $2^n$ length vector representation of $f$, i.e.,
$\mathbf{f}_{i} = f(i)$, $0 \leq i \leq 2^n-1$ and $H_n$ be the unitary matrix $H^{\otimes n}$, where $H = \frac{1}{\sqrt{2}} 
\left( \begin{array}{rr}
1 & 1 \\ 1 & -1
\end{array} \right)$.
The vector of real spectral coefficient of $f$ w.r.t. WHT 
is $\mathbf{f}^{*} = H_n \mathbf{f}$.
The weight of $f$ is then $|f| = \mathbf{f}_{0}^{*} \sqrt{2^{n}}$.

Ehrenfeucht and Karpinski \cite{ehrenfeucht90computational} designed an
algorithm, working in $\mathcal{O}(n^{3})$, for computing the weight of any
polynomial in $\mathbb{F}_{2}$ of degree at most $2$ with $n$ variables.
Given an arbitrary polynomial $f\in \mathbb{F}_{2}[x_{1},\ldots ,x_{n}]$ of
degree at most $2$, this algorithm computes a nonsingular $m\times n$ matrix
$T$ and a vector $\mathbf{c}$ of length $m$, such that%
\begin{equation}
g(T\mathbf{x}+\mathbf{c})=f(\mathbf{x}).  \label{eq:equivalentg}
\end{equation}%
The form of $g$ can be either

\begin{description}
\item[(Type I)]
\begin{equation*}
g=y_{1}+ y_{2}y_{3}+ y_{4}y_{5}+ \ldots + y_{m-1}y_{m}+ z;
\end{equation*}
\end{description}

or

\begin{description}
\item[(Type II)]
\begin{equation*}
g=y_{1}y_{2}+ y_{3}y_{4}+ \ldots + y_{m-1}y_{m}+ z,
\end{equation*}%
where $z\in \{0,1\}$.
\end{description}

In matrix form, we can write
\begin{equation}
g(\mathbf{y})=\mathbf{y}^{T}R_{m}\mathbf{y}+z,  \label{eq:matrixformg}
\end{equation}%
where $R_{m}$ is an $m\times m$ block diagonal matrix with the first block
being the identity matrix $I_{1}$ if and only if $g$ is of Type I and the
rest being $\lfloor m/2\rfloor $ blocks of the form $%
\begin{pmatrix}
0 & 1 \\
0 & 0%
\end{pmatrix}%
$. We call \emph{readonce form} a Boolean function of the same form of $g$,
since every variable $x_{i}$ appears exactly once in $g$. Hence, the
Ehrenfeucht-Karpinski algorithm transforms, by a nonsingular linear
substitution of variables, an arbitrary form $f\in \mathbb{F}%
_{2}[x_{1},\ldots ,x_{n}]$ into a readonce form $g\in \mathbb{F}%
_{2}[x_{1},\ldots ,x_{m}]$ ($m\leq n$), such that Eq. (\ref{eq:equivalentg})
holds.

\begin{definition}
\label{def:equivalentg} If $f$ is transformed into $g$ by the
Ehrenfeucht-Karpinski algorithm, we say that $g$ is the \emph{readonce form
equivalent} to $f$.
\end{definition}

Lidl and Niederreiter \cite{lidl} give us a closed formula to count the
weight of a readonce form:
\begin{equation}
|g|:=%
\begin{cases}
2^{m-1}, & \text{if $g$ is of Type I;} \\
2^{m-1}-(-1)^{z}2^{(m-2)/2}, & \text{if $g$ is of Type II.}%
\end{cases}
\label{cardg}
\end{equation}

The weight of an arbitrary function $f\in \mathbb{F}_{2}[x_{1},\ldots ,x_{n}]
$ of degree at most 2 is then
\begin{equation}
|f|=|g|2^{n-m}.  \label{cardf}
\end{equation}%
We can use this algorithm to calculate the MS-number of a graph. We are then
interested only in the weight of quadratic forms, which constitute a
subclass of the set of all functions to which we can apply the algorithm. If
we exclude the zero polynomial, the readonce forms equivalent to quadratic
forms have at least $2$ variables, $m\geq 2$.

For example, let us consider a graph $H$ with adjacency matrix%
\begin{equation*}
A(H)=%
\begin{bmatrix}
0 & 0 & 1 & 1 \\
0 & 0 & 1 & 1 \\
1 & 1 & 0 & 1 \\
1 & 1 & 1 & 0%
\end{bmatrix}%
.
\end{equation*}%
This is the complete graph on 4 vertices with a single missing edge. The
corresponding quadratic form is $%
f_{H}=x_{1}x_{3}+x_{1}x_{4}+x_{2}x_{3}+x_{2}x_{4}+x_{3}x_{4}.$ The readonce
form equivalent to $f_{H}(\mathbf{x})$ is $g_{H}(\mathbf{y})=y_{1}+y_{2}y_{3}
$, with the transformation
\begin{equation*}
\mathbf{y}=%
\begin{bmatrix}
1 & 1 & 0 & 0 \\
1 & 1 & 1 & 0 \\
1 & 1 & 0 & 1%
\end{bmatrix}%
\mathbf{x}.
\end{equation*}%
Since $m=3$, the MS-number of $H$ is $w(H)=8$. An implementation in \textsc{%
Octave} \cite{octave} of the algorithm to calculate the MS-number of a graph
is available 
at \url{http://arxiv.org/e-print/0906.2488} as a {\em tar} file.

In this paper we shall only consider connected graphs. For the disjoint
union of graphs we can in fact observe the following. Given two graphs $%
G_{1}=(V_{1},E_{1})$ and $G_{2}=(V_{2},E_{2})$ with disjoint vertex sets and
edge sets, their union $G=(V,E)=G_{1}\cup G_{2}$ is the graph with $%
V=V_{1}\cup V_{2}$ and $E=E_{1}\cup E_{2}$. Isolated vertices can be seen as
a particular case of graph union where $G_{1}$ or $G_{2}$ is an empty graph.
Obviously, the graph state corresponding to a disjoint union of graphs is
the tensor product of the states corresponding to the connected components.
If $G=\bar{K}_{n}\cup G_{2}$ is the union of an empty graph on $n$ nodes and
a graph $G_{2}$, then $w(G)=2^{n}w(G_{2})$. In general, for $G=G_{1}\cup
G_{2}$, we have $w(G)=w(G_{1})\bar{w}(G_{2})+\bar{w}(G_{1})w(G_{2})$. This
result can be recursively extended to the graph $G$, union of $n$ graphs $%
G_{i}$
\begin{equation*}
w(G)=w(\bigcup_{i=1}^{n}G_{i})=w(G_{1})\bar{w}(\bigcup_{i=2}^{n}G_{i})+\bar{w%
}(G_{1})w(\bigcup_{i=2}^{n}G_{i}).
\end{equation*}

\subsection{Readonce graph states}

The definition of \emph{readonce forms} leads us to introduce a special class of 
graph states, which can be written as
$$
\ket{G_{(m,z)}} = \frac{1}{\sqrt{2^{m}}}\sum_{y\in
\{0,1\}^{m}}(-1)^{g(y)}\ket{y},
$$
where $g$ is a readonce form in $m$ variables. 
We denote these states as \emph{readonce graph states}.
Given a number of qubits, we have only two different readonce graph states, 
since readonce forms depend only on two parameters. 
Also notice that the graphical structures beneath these states are perfect matchings.
Let us now remark that, with a graph state on $n$ qubits, we can associate a state of the following form:
$$
\ket{G_{(m,z)}} \otimes \ket{+}^{n-m},
$$
that is, a product of a readonce graph state on $m$ qubits and 
$n-m$ qubits in position $\ket{+}$.

For example, let us consider the complete graph $K_{3}$.
Its corresponding graph state can be prepared by the circuit depicted on the left of Fig. \ref{fig:readonce} (vertical segments are \emph{Controlled-}$\mathsf{Z}$ 
gates on two qubits).
The graph state $\ket{K_{3}}$ has the same MS-number of the readonce graph prepared
by the circuit shown in Fig. \ref{fig:readonce} (right).

\begin{figure}[ht]
\begin{minipage}[b]{0.4\linewidth}
\centering
\includegraphics[scale=1]{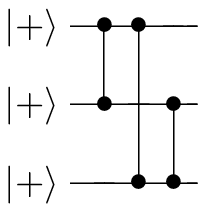}
\end{minipage}
\hspace{0.3cm}
\begin{minipage}[b]{0.4\linewidth}
\centering
\includegraphics[scale=1]{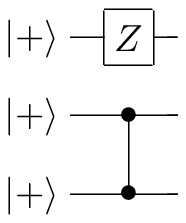}
\end{minipage}
\caption{Preparation of $\ket{K_{3}}$ (left) and of the readonce graph state associated to $\ket{K_{3}}$ (right).}
\label{fig:readonce}
\end{figure}

An important feature of readonce graph states is that all the gates act 
on the qubits in parallel, so their preparation cost is the lowest possible.

\section{Binary rank}

\label{sec:brank}

The \emph{binary rank} of a graph $G$, denoted by $\brank(G)$, is the rank
of its adjacency matrix calculated over $\mathbb{F}_{2}$.

\begin{lemma}[Theorem 8.10.1 in \protect\cite{agt}]
\label{lemma:brank}Let $A$ be an $n\times n$ adjacency matrix with binary
rank $m$. Then $m$ is even and there is an $m\times n$-matrix $C$ of rank $m$
such that $A=C^{T}N_{m}C$, where $N_{m}$ is a block diagonal matrix with $m/2
$ block $%
\begin{pmatrix}
0 & 1 \\
1 & 0%
\end{pmatrix}%
$.
\end{lemma}

The following theorem relates the MS-number of a graph $G$ to its binary rank:

\begin{theorem}
\label{th:brank} Let $f\in \mathbb{F}_{2}[x_{1},\ldots ,x_{n}]$ be the
quadratic form associated with a graph $G$. Let $g\in \mathbb{F}%
_{2}[y_{1},\ldots ,y_{m}]$ be the readonce form equivalent to $f$. If $g$ is
of Type I then $m=\brank(G)+1$. If $g$ is of Type II, $m=\brank(G)$.
\end{theorem}

\begin{proof}
From Eqs. (\ref{eq:frommatrix}), (\ref{eq:equivalentg}) and (\ref%
{eq:matrixformg}), we can write
\begin{equation*}
f_{G}(\mathbf{x})=\mathbf{x^{T}}U_{G}\mathbf{x}=g(T\mathbf{x}+\mathbf{c})=(T%
\mathbf{x}+\mathbf{c})^{T}R_{m}(T\mathbf{x}+\mathbf{c})+z,
\end{equation*}%
where $g$ is the readonce form equivalent to $f$ according to Definition \ref%
{def:equivalentg}. Since $f_{G}$ contains only quadratic terms, $%
U_{G}=T^{T}R_{m}T$. The adjacency matrix of $G$ can be then decomposed as
follows:
\begin{equation*}
A(G)=U_{G}+U_{G}^{T}=T^{T}(R_{m}+R_{m}^{T})T.
\end{equation*}%
If $g$ of Type II, this is the same decomposition of Lemma \ref{lemma:brank}%
, where $N_{m}=R_{m}+R_{m}^{T}$ and $C=T$. If $g$ of Type I, we consider the
submatrix $R_{s}$ obtained from $R_{m}+R_{m}^{T}$ by deleting the first row
and the first column. We also consider the matrix $T_{s}$ obtained from $T$
by deleting the first row. Since we deleted only zero vectors from $%
R_{m}+R_{m}^{T}$, we obtain $A(G)=T_{s}^{T}R_{s}T_{s}$. This is again the
decomposition of Lemma \ref{lemma:brank}, where $N_{m-1}=R_{s}$ and $C=T_{s}$%
. Therefore, $A(G)$ has binary rank $m-1$ if $g$ is of Type I and binary
rank $m$ if $g$ is of Type II.
\end{proof}

For instance, the graph $H$ in the above example has binary rank $\brank(H)=2
$ and its corresponding quadratic form is equivalent to the readonce form $%
g_{H}=y_{1}+y_{2}y_{3}$.

The binary rank of a graph is connected to bent properties of 
its corresponding Boolean function.
\begin{definition}
A Boolean function $f$ with even number of variables is \emph{bent} 
if its WHT spectrum is such that
\begin{itemize}
 \item $\mathbf{f}_{0}^{*} = 2^{(n-2)/2} \pm 2^{-1}$.
 \item $|\mathbf{f}_{i}^{*}| = 2^{-1}$ for any $i \neq 0$.
\end{itemize}
\end{definition}
Bent functions are important for cryptographic applications, since they have maximum distance from linear functions.

\begin{theorem}[Riera and Parker \cite{Riera05I}]
A graph $G$ has maximum binary rank if and only if 
its associated quadratic form $f_G$ is bent.
\end{theorem}

Readonce graph states constitute a class of graph states that correspond to 
bent quadratic Boolean functions.
In particular, among all the classes with this property, it is the one with minimum size.

\subsection{Bipartite graphs}

\label{sec:bip} A graph is \emph{$(A,B)$-bipartite} when its set of vertices
is partitioned in classes $A$ and $B$, such that each vertex in $A$ is
adjacent only to vertices in $B$ and \emph{viz}. In this section, we explore
the relation between MS-number and amount of entanglement in graph states
associated with bipartite graphs. A tool used for quantifying entanglement
in bipartite quantum systems is the \emph{Schmidt rank} \cite{Hein2007}. A
bipartite pure state $|\psi \rangle _{AB}$ can always be written in terms of
its Schmidt decomposition. This is a representation of $|\psi \rangle _{AB}$
in an orthogonal product basis, $|\psi \rangle _{AB}=\sum_{i}\lambda
_{i}|i_{A}\rangle \otimes |i_{B}\rangle $, where $\lambda _{i}\geq 0$. The
number of nonzero $\lambda _{i}$'s is called \emph{Schmidt rank} of $|\psi
\rangle _{AB}$. Let us consider a graph state $|G\rangle $, represented by
an \emph{$(A,B)$-bipartite} graph $G$. In this case, the Schmidt rank of $%
|G\rangle $, with respect to the bipartition $(A,B)$, is related to the
binary rank of $G$:
\begin{equation}
SR_{A}(G)=\frac{1}{2}\brank(G).  \label{eq:srank}
\end{equation}%
For bipartite graphs, we prove that the relation between MS-number and
binary rank is stronger than Theorem \ref{th:brank}. We first need a
technical lemma:

\begin{lemma}
\label{lemma:biprank} The quadratic form corresponding to a bipartite graph
with binary rank $m$ is equivalent to the readonce form $g = y_1y_2 + \ldots
+ y_{m-1}y_m$, that is, $g$ is always of Type II.
\end{lemma}

\begin{proof}
In order to prove this statement, we show how to construct a nonsingular
matrix $T$ that transforms the quadratic form $f_{G}$, associated with the
bipartite graph $G$, to its equivalent readonce form $g$ according to Eq. (%
\ref{eq:equivalentg}). Let $(A,B)$ be the bipartition of $G$ and let $p=|A|$
be the size of $A$. The adjacency matrix of $G$ can be decomposed as
\begin{equation*}
A(G)=%
\begin{pmatrix}
\mathbf{0} & A(G)_{AB} \\
A(G)_{AB}^{T} & \mathbf{0}%
\end{pmatrix}%
,
\end{equation*}%
being $A(G)_{AB}$ the matrix encoding edges between the classes. We assumed
that $\brank(G)=m$. So, the matrix $A(G)_{AB}$ has $m/2$ linearly
independent rows $\gamma _{1},\gamma _{2},\ldots \gamma _{m/2}$. We can
write the $j$th row of $A(G)_{AB}$ as a linear combination of $\gamma
_{1},\gamma _{2},\ldots \gamma _{m/2}$: $A(G)_{AB_{j}}=x_{1,j}\gamma
_{1}+\ldots +x_{m/2,j}\gamma _{m/2}$, for $1\leq j\leq p$. The matrix $T$ is
then composed of the following $m$ rows:
\begin{equation*}
\begin{array}{llccr}
T_{2i-1} & =(x_{i,1} & \ldots  & x_{i,p} & \mathbf{0}) \\
T_{2i} & =( & \mathbf{0} &  & \gamma _{i})%
\end{array}%
\qquad 1\leq i\leq m/2.
\end{equation*}%
The observation that $T$ is nonsingular proves the lemma.
\end{proof}

In the light of Lemma \ref{lemma:biprank} and Eq. (\ref{eq:srank}) we obtain
a direct relation between MS-number and Schmidt rank of bipartite graph
states:

\begin{theorem}
\label{th:ms_bip} Let $\vert G\rangle$ be a graph state represented by a $%
(A,B)$-bipartite graph $G$ of order $n$ and let $r = SR_{A}(\vert G\rangle)$
be its Schmidt rank with respect to the bipartition of $G$. Then the
MS-number of $G$ is $2^{n-1}(1-2^{-r})$.
\end{theorem}

\subsection{Edge-local complementation}

We characterize how the weight of quadratic forms changes under the
operation of edge-local complementation on bipartite graphs. The definition
of edge-local complementation is derived from the definition of local
complementation, an operation that plays a fundamental role in the context
of graph states \cite{Ji2007}. The operation of edge-local complementation
itself is directly related to the action of local Hadamard transformations
on graph states \cite{Nest2007}. The \emph{local complement} $G^{v}$ of a
graph $G=(V,E)$ at one of its vertices $v\in V(G)$ is the graph obtained by
complementing the subgraph of $G$ induced by the neighborhood $N_{v}$ of $v$
and leaving the rest of $G$ unchanged. Given a graph $G=(V,E)$ and an edge $%
\{u,v\}\in E(G)$, \emph{edge-local complementation} (or \emph{pivoting}) on $%
\{u,v\}$ transforms $G$ into $G^{(uv)}=((G^{u})^{v})^{u}=((G^{v})^{u})^{v}$.

This operation finds application also in coding theory. In particular, the
pivot orbit of a bipartite graph corresponds to the equivalence class of a
binary linear code \cite{Danielsen2007}. A graph $H$ is a \emph{pivot minor}
of $G$ if $H$ can be obtained from $G$ by vertex deletions and pivotings. We
saw in Theorem \ref{th:ms_bip} that the MS-number and binary rank of
bipartite graphs are strictly related. By our result together with Theorem
8.10.2 in \cite{agt}, we can state the next observation:

\begin{proposition}
Let $G$ be a bipartite graph and let $G^{(uv)}$ be the edge-local complement
of $G$ on the edge $(u,v)$. If $G^{\prime }$ is the pivot-minor of $G$
obtained by deleting $u$ and $v$ from $G^{(uv)}$, then $G^{\prime }$ is also
bipartite and $w(G)=2^{n-2}+2w(G^{\prime })$.
\end{proposition}

\section{General properties}

\label{sec:properties}In this section, we list some properties of the
MS-number for generic graphs. First of all, obviously, the MS-number of a
graph is zero if and only if the graph is empty.

\begin{proposition}
For an arbitrary graph $G=(V,E)$, we have $w(G)\geq |E|$.
\end{proposition}

\begin{proof}
For each edge let us consider the subgraph induced only by the endpoints of
the edge. This subgraph has clearly odd size.
\end{proof}

\begin{proposition}
If we exclude the path graph $P_2$, it holds that $w(G)$ is even for any $G$.
\end{proposition}

\begin{proof}
It is clear from Eqs. (\ref{cardg}) and (\ref{cardf}) that $w(G) = |f_{G}|$
is odd if and only if $n = |V(G)| = 2$ and $|E(G)| > 0$, that is, if and
only if $G = P_2$.
\end{proof}

Further lower and upper bounds can be easily given:

\begin{proposition}
\label{prop:minimum} For each graph $G$ of order $n$ and size $m>0$, we have
$w(G)\geq 2^{n-2}$.
\end{proposition}

\begin{proof}
The weight of $f_{G}$ is minimum when its equivalent readonce form is $%
g=y_{1}y_{2}$. This is the case of complete bipartite graphs. A \emph{%
complete bipartite graph} is a special kind of $(A,B)$-bipartite graph in
which every vertex in $A$ is connected to every vertex in $B$ and \emph{viz}%
. If the two classes have cardinality $p$ and $q$, then the graph is denoted
by $K_{p,q}$. It follows that  $w(K_{p,q})=2^{p+q-2}$.
\end{proof}

\begin{proposition}
For each order $n\geq 4$, the graph $Q_{n}=(K_{4}\cup \bar{K}_{n-4})$ has
maximum MS-number. That is, given an arbitrary graph $G$ of order $n$, it
holds: $w(G)\leq w(Q_{n})$, for $n\geq 4$.
\end{proposition}

\begin{proof}
The MS-number of the graph $Q_{n} = (K_4 \cup \bar{K}_{n-4})$ is $w(Q_{n}) =
2^{n-1}+2^{n-3}$. Looking at Eq. (\ref{cardg}), it is clear that $w(G)$ is
maximum when $g$ is of type II and $z=1$. Under these conditions, the
minimum value of $m$ is $4$. Therefore, $w(G) = 2^{n-1} (1 + 2^{-m/2})$ is
maximum when $m=4$, that is, when $w(G) = w(Q_{n})$.
\end{proof}

\section{Examples}

\label{sec:examples}In this section, we give explicit formulas for the
MS-number for some familiar classes of graphs.

\begin{proposition}
For a complete graph $K_n$ of order $n \geq 1$, it holds that
\begin{equation*}
w(K_n) = \sum_{i=0}^{\lfloor(n-1)/4\rfloor}\binom{n+1}{4i+3}.
\end{equation*}
\end{proposition}

\begin{proof}
We know from graph theory that any subgraph induced by a clique is a
complete subgraph and that the number of edges in a complete graph $K_{n}$
of order $n$ is equal to the $n$th triangular number $t_{n}=\sum_{i=1}^{n}i$%
. It is also easy to show that the sequence of triangular numbers goes on
according to the pattern odd-odd-even-even-odd-odd and so on. We can then
count the number of the subgraphs of $K_{n}$ with odd size and claim that
\begin{align*}
w(K_{n})& =\sum_{i=0}^{\lfloor (n-2)/4\rfloor }\left( \binom{n}{4i+2}+\binom{%
n}{4i+3}\right)  \\
& =\sum_{i=0}^{\lfloor (n-1)/4\rfloor }\binom{n+1}{4i+3}.
\end{align*}
Notice that a quadratic form corresponding to a complete graph is a 
\textit{symmetric} Boolean function, that is, 
its output depends only on the number of unit values among the input variables. 
\end{proof}

\begin{proposition}
\label{prop:path}Let $P_{n}$ be the path graph of order $n$. For $n\geq 1$,
\begin{equation*}
w(P_{n})=%
\begin{cases}
2^{n-1}-2^{(n-1)/2}, & \text{if $n$ is odd;} \\
2^{n-1}-2^{(n-2)/2}, & \text{if $n$ is even.}%
\end{cases}%
\end{equation*}
\end{proposition}

\begin{proof}
A path graph on $n$ vertices corresponds to the function $%
f_{P_{n}}=\sum_{j=1}^{n-1}x_{j}x_{j+1}$. If $n$ is odd then $f_{P_{n}}$ is
equivalent to the readonce function with $n-1$ variables and $z=0$ under the
following transformation of variables:
\begin{equation*}
\text{for }i\leq n-1,\qquad y_{i}=%
\begin{cases}
x_{i}+x_{i+2}, & \text{if $i$ is odd;} \\
x_{i}, & \text{if $i$ is even.}%
\end{cases}%
\end{equation*}%
Since $n-1$ is even, $w(P_{n})=2(2^{(n-1)-1}-2^{[(n-1)-2]/2}
) = 2^{n-1}-2^{(n-1)/2}$. If $n$ is even then the function $f_{P_{n}}$
corresponds to the readonce function with $n$ variables and $z=0$. In this
case, the transformation of variables is
\begin{equation*}
\begin{array}{l}
y_{1}=x_{2}; \\
\text{for }2\leq i\leq n,\qquad y_{i}=%
\begin{cases}
x_{i-1}+x_{i+1}, & \text{if $i$ is odd;} \\
x_{i-1}, & \text{if $i$ is even.}%
\end{cases}%
\end{array}%
\end{equation*}%
Then, $w(P_{n})=2^{n-1}-2^{(n-2)/2}$.
\end{proof}

\begin{proposition}
Let $C_n$ be the cycle graph of order $n$. For $n \geq 2$,
\begin{equation*}
w(C_n) =
\begin{cases}
2^{n-1}, & \text{if $n$ is odd;} \\
2^{n-1}- 2^{n/2}, & \text{if $n$ is even.}%
\end{cases}%
\end{equation*}
\end{proposition}

\begin{proof}
A cycle graph on $n$ vertices corresponds to the quadratic form $%
f_{C_{n}}=\sum_{j=1}^{n-1}x_{j}x_{j+1}+x_{n}x_{1}$. If $n$ is odd then  $%
f_{P_{n}}$ is equivalent to the readonce function with $n$ variables and $z=0
$. Hence, $w(C_{n})=2^{(n-1)}$. If $n$ is even then the function $f_{P_{n}}$
corresponds to the readonce function with $n-2$ variables and $z=0$. Hence, $%
w(C_{n})=2^{2}(2^{n-2-1}-2^{(n-2-2)/2})=2^{n-1}-2^{n/2}$.
Transformations of variables are similar to the ones used in the proof of
Proposition \ref{prop:path}.
\end{proof}

\begin{proposition}
If $S_{n}$ is a star graph of order $n\geq 2$ then $w(S_{n})=2^{n-2}$.
\end{proposition}

\begin{proof}
The function that represents the star graph $S_{n}$ is $f_{S_{n}}=%
\sum_{j=2}^{n}x_{1}x_{j}$. Its equivalent readonce function is $%
g_{S_{n}}=y_{1}y_{2}$ with the transformation $\{y_{1}=x_{1},y_{2}=x_{2}+%
\ldots +x_{n}\}$. It follows that $w(S_{n})=2^{n-2}$.
\end{proof}

Star graphs belong to the more general class of complete bipartite graphs
(see Proposition \ref{prop:minimum}).

A \emph{tree} is a graph without cycles. Every tree is a bipartite graph.
The Schmidt rank of a tree coincides with its vertex covering number \cite%
{Hein2007}. From Theorem \ref{th:ms_bip}, we can evaluate the MS-number of
trees:

\begin{proposition}
If $T$ is a tree with vertex covering number $\tau $ then $%
w(T)=2^{n-1}(1-2^{-\tau })$.
\end{proposition}

In Fig. \ref{tablems}, we show classes of graphs with same order ($|V| \leq 9$) and MS-number. Graphs have been computed using \textsc{nauty} \cite{nauty}.

\section{Conclusions}

\label{sec:conclusion}

We have given a proof of a simple relation
between the binary rank of a graph and the weight of the corresponding
quadratic forms. For bipartite graphs, this means a connection between
Schmidt-rank and weight of quadratic forms. This led to a straightforward
characterization of the weight of functions associated with pivot-minor of
bipartite graphs. 
We believe that further combinatorial analysis on Boolean
functions can be made in order to achieve a better classification of graph
states.

\begin{acknowledgments}
A.C. is especially grateful to Anna Bernasconi 
for her guidance during the early stage of this paper.
We thank Niel de Beaudrap, Zhengfeng Ji and Matthew G. Parker for
answering with very informative letters to our questions. We are grateful to
Marc Thurley for alerting us about the existence of the Ehrenfeucht-Karpinski
algorithm.
While writing this paper S.S. was with the Institute for Quantum Computing and the Department of Combinatorics and Optimization at the University of Waterloo. 
Research at the Institute for Quantum Computing is supported by DTOARO, ORDCF, CFI, CIFAR, and
MITACS.
\end{acknowledgments}

\newpage
\begin{figure*}[htb]
 \caption{\label{tablems} Classes of graphs with the same MS-number up to isomorphism.
 The table shows the MS-number of each class and its cardinality
 (the number in parenthesis).
 We take the representatives with the smallest number of edges.
 }
\bigskip
 \centering
 \includegraphics[bb=15 247 544 759]{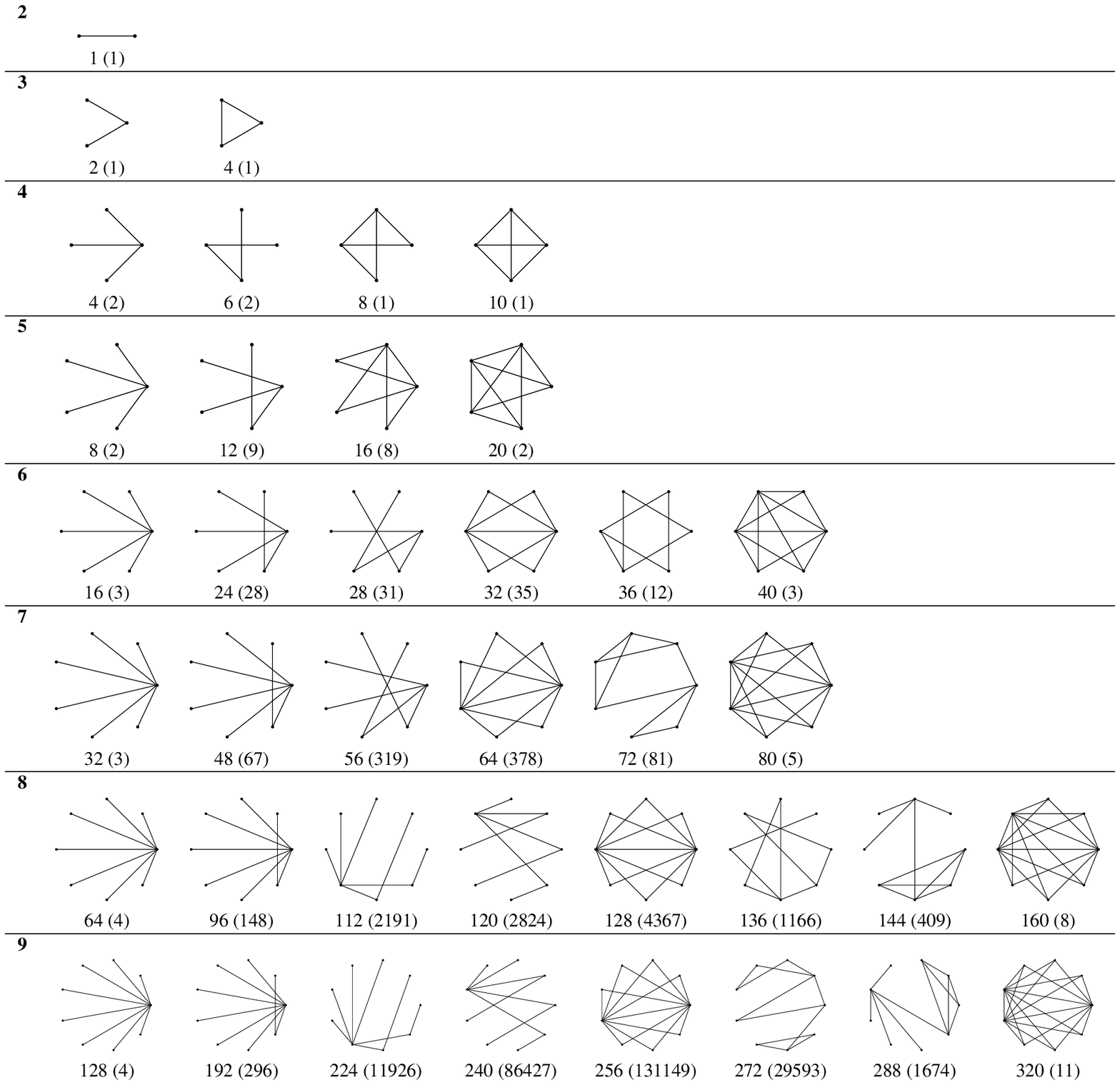}
\end{figure*}

\end{document}